\documentclass[reqno]{llncs}

\usepackage{times,amsmath,amssymb,enumerate}

\newcommand{\CWP}{\mathsf{CWP}}
\newcommand{\coNP}{\mathsf{coNP}}
\newcommand{\coRP}{\mathsf{coRP}}
\newcommand{\coRNC}{\mathsf{coRNC}}
\newcommand{\diff}{\Delta}
\newcommand{\dom}{\mathop{\mathrm{dom}}}

\newcommand{\graph}{\mathsf{graph}}

\newcommand{\univ}{U}
\newcommand{\NC}{\mathsf{NC}}

\newcommand{\poly}{\mathsf{val}}
\newcommand{\Ptime}{\mathsf{P}}

\newcommand{\rhs}{\mathsf{rhs}}

\newcommand{\RNC}{\mathsf{RNC}}

\newcommand{\size}{\mathsf{size}}

\newcommand{\TC}{\mathsf{TC}}

\newcommand{\val}{\mathrm{val}}

\newcommand{\ms}{|\! |}

\newcommand{\dA}{\mathbb{A}}
\newcommand{\dB}{\mathbb{B}}

\newcommand{\cC}{\mathcal{C}}

\begin{document}

\title{Parallel Identity Testing for Skew Circuits with Big Powers and Applications}

\author{Daniel K\"onig \and Markus Lohrey}

\institute{Universit\"at Siegen, Germany\\
  \texttt{\{koenig,lohrey\}@informatik.uni-leipzig.de} }

\maketitle

\begin{abstract}
Powerful skew arithmetic circuits are introduced. These are 
skew arithmetic circuits with variables, where input gates can be labelled with powers $x^n$ for 
binary encoded numbers $n$. It is shown that polynomial identity testing for powerful skew arithmetic circuits
belongs to $\coRNC^2$, which generalizes a corresponding result for (standard) skew circuits.
Two applications of this result are presented: (i) Equivalence of higher-dimensional straight-line programs
can be tested in $\coRNC^2$; this result is even new in the one-dimensional case, where the straight-line programs
produce strings. (ii) The compressed word problem (or circuit evaluation problem) for certain wreath products 
of finitely generated abelian groups belongs to $\coRNC^2$.
\end{abstract}

\section{Introduction}

\emph{Polynomial identity testing} is the following computational problem:
The input is a circuit, whose internal gates are labelled with either addition or multiplication
and its input gates are labelled with variables ($x_1, x_2, \ldots$) or constants ($-1,0,1$), and
it is asked whether the output gate evaluates to the zero polynomial (in this paper, we always
work in the polynomial ring over the coefficient ring $\mathbb{Z}$ or $\mathbb{Z}_n$ 
for  $n \geq 2$). Based on the Schwartz-Zippel-DeMillo-Lipton Lemma,
Ibarra and Moran \cite{IbMo83} proved that polynomial identity testing over $\mathbb{Z}$ or $\mathbb{Z}_p$ belongs to the class {\sf coRP}
(the complements of problems in randomized polynomial time). Whether there is a deterministic
polynomial time algorithm for polynomial identity testing is an important problem.
In \cite{ImpWig97} it is shown that if there exists a language in $\mathsf{DTIME}(2^{\mathcal{O}(n)})$
that has circuit complexity $2^{\Omega(n)}$, then $\mathsf{P} = \mathsf{BPP}$ (and hence
$\mathsf{P} = \mathsf{RP} = \mathsf{coRP}$). There is also an implication that goes the other way round:
 Kabanets and Impagliazzo \cite{KabanetsI04} proved that if polynomial identity testing
belongs to $\mathsf{P}$, then (i) there is a language in $\mathsf{NEXPTIME}$
 that does not have polynomial size circuits, or (ii)
the permanent is not computable by polynomial size arithmetic circuits.
Both conclusions represent major open problems in complexity theory. Hence, although
it is quite plausible that polynomial identity testing belongs to $\mathsf{P}$ (by \cite{ImpWig97}),
it will be probably very hard to prove (by \cite{KabanetsI04}).

It is known that for algebraic formulas (where the circuit is a tree) and 
more generally, skew circuits (where for every multiplication gate, one of the 
two input gates is a constant or a variable),
polynomial identity testing belongs to $\coRNC$ (but it is still not known to be in $\Ptime$), see \cite[Corollary~2.1]{KabanetsI04}. This holds,
since algebraic formulas and skew circuits can be evaluated in $\NC$ if the variables
are substituted by concrete (binary coded) numbers. Then, as for general polynomial identity testing, 
the Schwartz-Zippel-DeMillo-Lipton Lemma yields a $\coRNC$-algorithm.

In this paper, we identify a larger class of  circuits, for which polynomial identity testing is still
in $\coRNC$; we call these circuits {\em powerful skew circuits}. In such a circuit, we require that for 
every multiplication gate, one of the two input gates is either a constant or a power $x^N$ of a variable $x$,
where the exponent $N$ is given in binary notation. One can replace this power $x^N$ by a subcircuit 
of size $\log N$ using iterated squaring, but the resulting circuit is no longer skew. The main result
of this paper states that polynomial identity testing for powerful skew circuits over the rings $\mathbb{Z}[x]$ and 
$\mathbb{F}_p[x]$ is still in $\coRNC$ (in fact, $\coRNC^2$).
For this, we use an identity testing algorithm of Agrawal and Biswas \cite{AgrawalB03}, which computes the 
output polynomial of the circuit modulo a polynomial $p(x)$ of polynomially bounded degree, which is 
randomly chosen from a certain sample space.
Moreover, in our application, all computations can be done in the ring $\mathbb{F}_p[x]$ for 
a prime number $p$  of polynomial size. This allows us to compute the big powers $x^N$ modulo
$p(x)$ in $\NC^2$ using an algorithm of Fich and Tompa \cite{FiTo85}. It should be noted that the application
of the Agrawal-Biswas algorithm is crucial in our situation. If, instead we would use the 
Schwartz-Zippel-DeMillo-Lipton Lemma, then we would be forced to compute $a^N \mod m$ for 
randomly chosen numbers $a$ and $m$ with polynomially many bits. Whether this problem 
(modular powering) belongs to $\NC$ is a famous open problem \cite[Problem B.5.6]{GrHoRu95}.

We present two applications of our $\coRNC$ identity testing algorithm. The first one concerns
the equivalence problem for straight-line programs. Here, a straight-line program (SLP) is a 
context-free grammar $G$ that computes a single word $\val(G)$. In this context, SLPs are
extensively used in data compression and algorithmics on compressed data, see \cite{Loh12survey}
for an overview. It is known that equivalence for SLPs, i.e., the question whether $\val(G) = \val(H)$
for two given SLPs, can be decided in polynomial time. This result was independently  
discovered by  Hirshfeld, Jerrum, and Moller \cite{HiJeMo96},  Mehlhorn, Sundar, and Uhrig \cite{MehlhornSU97}, and
Plandowski \cite{Pla94}. All known algorithms for the equivalence test are sequential and it is not clear
how to parallelize them. Here, we exhibit an $\NC^2$-reduction from the equivalence problem for SLPs to 
identity testing for skew powerful circuits. Hence, equivalence for SLPs belongs to $\coRNC$.
Moreover, our reduction immediately generalizes to higher dimensional pictures for which SLPs can 
be defined in a fashion similar to the one-dimensional (string) case, using one concatenation operation
in each dimension. For two-dimensional SLPs, Berman et al.~\cite{BeKaLaPlRy02} proved that equivalence
belongs to $\coRP$ using a reduction to PIT. We can improve this result to $\coRNC$. Whether equivalence
of two-dimensional (resp., one-dimensional) SLPs belongs to {\sf P} (resp., $\NC$) is open.

Our second application concerns the compressed word problem for groups. Let $G$ be a finitely 
generated (f.g.) group, and let $\Sigma$ be a finite generating set for $G$. For the compressed word
problem for $G$, briefly $\CWP(G)$, the input is an SLP (as described in the preceding paragraph) over the alphabet
$\Sigma \cup \Sigma^{-1}$, and it is asked whether $\val(G)$ evaluates to the group identity.
The compressed word problem is a succinct version of the classical
word problem (Does a given word over $\Sigma \cup \Sigma^{-1}$ evaluate to the group identity?).
One of the main motivations for the compressed word problem is the fact that the classical word problem
for certain groups (automorphism groups, group extensions) can be reduced to the compressed
word problem for simpler groups \cite[Section~4.2]{Loh14}.
For finite groups (and monoids) the compressed word problem was studied in Beaudry et al.~\cite{BeMcPeTh97},
and for infinite groups the problem was studied for the first time in \cite{Loh06siam}. Subsequently, 
several important classes of f.g. groups with polynomial time compressed word problems were found:  
f.g. nilpotent groups,  f.g. free groups, graph groups (also known as right-angled Artin groups or partially commutative groups), and 
virtually special groups. The latter contain all Coxeter groups, one-relator groups with torsion,
fully residually free groups, and fundamental groups of hyperbolic 3-manifolds; see \cite{Loh14} for details.
For the important class of f.g. linear groups, i.e., f.g. groups of matrices over a field, one can show that the 
compressed word problem reduces to polynomial identity testing  (over $\mathbb{Z}$
or $\mathbb{Z}_p$, depending on the characteristic of the field) and hence belongs to $\mathsf{coRP}$
\cite[Theorem~4.15]{Loh14}.
Vice versa, it was shown that polynomial identity testing over $\mathbb{Z}$ can be reduced
to the compressed word problem for the linear group $\mathsf{SL}_3(\mathbb{Z})$ \cite[Theorem~4.16]{Loh14}. The proof is based on 
a construction of Ben-Or and Cleve  \cite{Ben-OrC92}.  This result indicates that derandomizing the compressed word problem
for a f.g. linear group will be in general very difficult.

In this paper, we consider the compressed word problem for wreath products. If $G$ is a f.g. non-abelian group, then
the compressed word problem for the wreath product $G \wr \mathbb{Z}$ is $\coNP$-hard 
\cite[Theorem~4.21]{Loh14}. On the other hand, we prove that $\CWP(\mathbb{Z} \wr \mathbb{Z})$ is equivalent w.r.t. $\NC^2$-reductions
to identity testing for powerful skew circuits. In particular, $\CWP(\mathbb{Z} \wr \mathbb{Z})$ belongs to $\coRNC$.
The latter result generalizes to any wreath product $G \wr H$, where $H = \mathbb{Z}^n$ for some $n$ and 
$G$  is a finite direct product of copies of $\mathbb{Z}$ and $\mathbb{Z}_p$ for primes $p$.

\section{Background from complexity theory}

Recall that $\mathsf{RP}$ is the set of all problems
$A$ for which there exists a polynomial time bounded randomized Turing machine $R$ such that: (i) if 
$x \in A$ then $R$ accepts $x$ with probability at least $1/2$, and (ii) if $x \not\in A$ then $R$ accepts
$x$ with probability $0$. The class $\mathsf{coRP}$ is the class of all complements of problems from $\mathsf{RP}$.

We use standard definitions concerning circuit complexity, see e.g. \cite{Vol99} for more details. 
In particular we will consider the class $\NC^i$ of all problems
that can be solved by a circuit family $(\mathcal{C}_n)_{n \geq 1}$, where the size of 
$\mathcal{C}_n$ (the circuit for length-$n$ inputs) is polynomially bounded in $n$, its depth is bounded by $O(\log^i n)$,
and $\mathcal{C}_n$ is built from input gates, NOT-gates and  AND-gates and OR-gates of fan-in two.
The class $\NC$ is the union of all classes $\NC^i$.
All circuit families in this paper  will be logspace-uniform, which means  that the 
mapping $a^n \mapsto \mathcal{C}_n$ can be computed in logspace.
A few times, we will mention the class {\sf DLOGTIME}-uniform $\TC^0$, see \cite{HeAlBa02} for details.
Here, it is only important that {\sf DLOGTIME}-uniform $\TC^0$ is contained in $\NC^1$.

To define a randomized version of $\NC^i$, one uses circuit families with additional inputs. 
So, let the $n^{\text{th}}$ circuit $\mathcal{C}_n$ in the family have $n$ normal input gates plus $m$ random input gates,
where $m$ is polynomially bounded in $n$. For an input $x \in \{0,1\}^n$ one defines 
the acceptance probability as
$$
\mathsf{Prob}[\mathcal{C}_n \text{ accepts } x] = \frac{ |\{ y \in \{0,1\}^m \mid \mathcal{C}_n(x,y) = 1 \}|}{2^m}
$$
Here, $\mathcal{C}_n(x,y) = 1$ means that the circuit $\mathcal{C}_n$ evaluates to $1$ if the $i^{\text{th}}$ normal
input gate gets the $i^{\text{th}}$ bit of the input string $x$, and  the $i^{\text{th}}$ random input gate
gets the $i^{\text{th}}$ bit of the random string $y$. Then, the class $\RNC^i$ is the class of all problems
$A$ for which there exists a polynomial size circuit family $(\mathcal{C}_n)_{n \geq 0}$ of depth $O(\log^i n)$ with random input gates
that uses NOT-gates and  AND-gates and OR-gates of fan-in two, such that for all inputs $x \in \{0,1\}^*$ of length $n$:
(i) if $x \in A$, then $\mathsf{Prob}[\mathcal{C}_n \text{ accepts } x] \geq 1/2$, and
(ii) if $x \not\in A$, then $\mathsf{Prob}[\mathcal{C}_n \text{ accepts } x] = 0$.
As usual, $\coRNC^i$ is the class of all complements of problems from $\RNC^i$.
Section B.9 in \cite{GrHoRu95} contains several problems that are known to be 
in $\RNC$, but which are not known to be in $\NC$; the most prominent example is the existence of a perfect matching in a graph.

\section{Polynomials and circuits}

n this paper we deal with polynomial rings $R[x_1, \ldots, x_k]$ in several variables, where 
$R$ is the ring of integers $\mathbb{Z}$ or the ring $\mathbb{Z}_n$ of integers modulo $n \geq 2$.
For computational problems,  we  have to distinguish between two representations of polynomials. Let
$$
p(x_1, \ldots, x_k) = \sum_{i=1}^l a_i x_1^{e_{i,1}} \cdots x_k^{e_{i,k}}
$$
be a multivariate polynomial.
\begin{itemize}
\item The {\em standard representation}  of $p(x)$
is the sequence of tuples $(a_i, e_{i,1}, \ldots, e_{i,k})$, where the coefficient $a_i$ is represented
in binary notation (of course this is only important for the coefficient ring $\mathbb{Z}$) and the exponents 
$e_{i,j}$ are represented in unary notation. Let $|p| =   \sum_{i=1}^n (\lceil \log |a_i| \rceil + e_{i,1} + \cdots + e_{i,k})$.
\item  The {\em succinct representation} of $p(x)$ is the 
sequence of tuples $(a_i, e_{i,1}, \ldots, e_{i,k})$, where both the coefficient $a_i$ 
and the exponents $e_{i,j}$ are represented in binary notation.
Let $\ms p \ms =   \sum_{i=1}^n (\lceil \log |a_i| \rceil + \lceil \log e_{i,1} \rceil + \cdots + \lceil \log e_{i,k} \rceil)$.
\end{itemize}
We use the following result of Eberly \cite{Eberly89} (see also \cite{HeAlBa02}).

\begin{proposition} \label{prop-eberly}
Iterated addition,  iterated multiplication, and division with remainder of polynomials from $\mathbb{Z}[x]$
or $\mathbb{F}_p[x]$ ($p$ is a prime that can be part of the input in binary encoding)
 that are given in standard representation
belong to $\NC^1$ (in fact, $\mathsf{DLOGTIME}$-uniform $\TC^0$).
\end{proposition}
Consider a commutative semiring $\mathcal{S} = (S,\oplus, \otimes)$. An algebraic circuit (or just circuit) over $\mathcal{S}$
is a triple  $\cC = (V, \rhs, A_0)$, where $V$ is a finite set of
{\em gates} or {\em variables}, $A_0 \in V$ is the output gate,
and $\rhs$ (for {\em right-hand side}) maps every $A \in V$ to an expression (the right-hand side of $A$) 
of one  of the following three forms:
\begin{itemize}
\item a semiring element $s \in S$ (such a gate is an {\em input gate}),
\item $B \oplus C$ with $B,C \in V$ (such a gate is an {\em addition gate}),
\item $B \otimes C$ with $B,C \in V$ (such a gate is a {\em multiplication gate}).
\end{itemize}
Moreover, we require that the directed graph
$$
\graph(\cC) = (V, \{ (A,B) \in V \times V \mid \text{ $B$ occurs in $\rhs(A)$} \})
$$
is acyclic. Every gate $A \in V$ evaluates to an element $\val_{\cal C}(A) \in S$ in the natural
way and we set $\val({\cal C}) = \val_{\cal C}(A_0)$. A circuit over $\mathcal{S}$
 is called skew if for every multiplication gate $A$ one of the two gates  (or both of them) in $\rhs(A)$ is an
 input gate. 
 
 A branching program over $\mathcal{S}$ is a tuple $\mathcal{A} = (V,E, \lambda, s,t)$, 
 where $(V,E)$ is a directed acyclic graph, $\lambda : E \to S$ assigns to each
 edge a semiring element, and $s,t \in V$. Let $\mathcal{P}$ be the set of all paths
 from $s$ to $t$. For a path $p = (v_0, v_1, \ldots, v_n) \in \mathcal{P}$ ($v_0 = s$, $v_n = t$)
 we define $\lambda(p) = \prod_{i=1}^n \lambda(v_{i-1},v_i)$ as  the product (w.r.t. $\otimes$)
 of all edge labels along the path. Finally, the value defined by $\mathcal{A}$
 is 
 $$
 \val(\mathcal{A}) = \sum_{p \in \mathcal{P}} \lambda(p) . 
 $$
 It is well known that skew circuits and branching programs are basically the same objects.

It is well known that
the value defined by a branching program $\mathcal{A}$ can be computed using matrix powers.
W.l.o.g.~assume that   $\mathcal{A} = (\{1,\ldots,n\},E, \lambda, 1,n)$ and consider the 
adjacency matrix $M$ of the edge-labelled graph $(\{1,\ldots,n\},E, \lambda)$, i.e., the $(n \times n)$-matrix 
$M$ with $M[i,j] = \lambda(i,j)$. Then 
$$
\val(\mathcal{A}) = \bigg(\sum_{i=0}^n M^i \bigg)[1,n] .
$$ 
For many semirings $\mathcal{S}$, this simple fact can be used to get an $\NC^2$-algorithm for computing
$\val(\mathcal{A})$. The $n+1$ matrix powers $M^i$ ($0 \leq i \leq n$) can be computed in parallel, and every power can be 
computed by a balanced tree of height $\log i \leq \log n$, where every tree node computes a matrix product.
Hence, we obtain an $\NC^2$-algorithm, if 
\begin{enumerate}[(i)]
\item the number of bits needed to represent a matrix entry in $M^n$ 
is polynomially bounded in $n$ and the number of bits of the entries in $M$, and 
\item the product of two matrices over the semiring $\mathcal{S}$ can be computed in $\NC^1$. 
\end{enumerate}
Point (ii) holds if products of two elements and iterated sums in $\mathcal{S}$
can be computed in $\NC^1$. For the following important semirings these facts are well known (see also Proposition~\ref{prop-eberly}): 
$(\mathbb{Z}[x], +, \cdot)$,  $(\mathbb{Z}_n[x], +, \cdot)$ for $n \geq 2$, $(\mathbb{Z} \cup \{\infty\}, \min, +)$, and
$(\mathbb{Z} \cup \{-\infty\}, \max, +)$.
Here, we assume that polynomials are given in the {\em standard representation}. For the polynomial
ring $\mathbb{Z}[x]$ also note that every entry $p(x)$ of the matrix power $M^n$ is a polynomial of degree $n \cdot m$,
where $m$ is the maximal degree of a polynomial in $M$, and all coefficients  are bounded by 
$(n \cdot m \cdot a)^n$ (and hence need
at most $n\cdot (\log n + \log m  + \log n)$ bits), where $a$ is the maximal absolute value of a coefficient in $M$. Hence point (i) above holds.
The following lemma sums up the above discussion.

\begin{lemma} \label{lemma-evaluate-skew}
The output value of a given skew circuit (or branching program) over one of the following semirings can be computed in $\NC^2$:
\begin{enumerate}[(i)]
\item $(\mathbb{Z}[x], +, \cdot)$ and  $(\mathbb{Z}_n[x], +, \cdot)$ for $n \geq 2$ (polynomials are given in the 
standard representation, and $n$ can be part of the input in binary representation)
\item $(\mathbb{Z} \cup \{\infty\}, \min, +)$ and $(\mathbb{Z} \cup \{-\infty\}, \max, +)$ (integers are given in binary representation)
\end{enumerate} 
\end{lemma}
Point (i) of Lemma~\ref{lemma-evaluate-skew} also holds for the polynomial rings $(\mathbb{Z}[x_1, \ldots, x_k], +, \cdot)$ and $(\mathbb{Z}_n[x_1, \ldots, x_k], +, \cdot)$
as long as the number $k$ of variables is not part of the input: The polynomial $p(x_1, \ldots, x_k) = \prod_{i=1}^k (x_i+1)$ can be defined
by a branching program with $O(k)$ edges labeled by the polynomials $x_i+1$, but the product of these polynomials has $2^k$ monomials.
Also note that it is important that we use the standard representation for polynomials in (i):
The polynomial $p(x) = \prod_{i=1}^n (x^{2^i} +1)$ can be represented by a  branching program with $O(n)$ edges labeled by the polynomials $x^{2^i}+1$
but $p(x)$ has $2^n$ monomials.

In this paper, we will mainly deal with circuits over a polynomial ring $R[x_1, \ldots, x_k]$, where the ring
$R$ is either $(\mathbb{Z}, +, \cdot)$ or $(\mathbb{Z}_n, +, \cdot)$. Let $R$ be one of these rings.
By definition, in such a circuit every input gate is labelled
with a polynomial from $R[x_1, \ldots, x_k]$. Usually, one considers circuits where the right-hand side
of an input gate is a polynomial given in standard representation (or, equivalently, a constant $a \in R$ or variable $x_i$);
we will also use the term ``standard circuits'' in this case. 
For succinctness reasons, we will also consider circuits over $R[x_1, \ldots, x_k]$, where 
the right-hand sides of input gates are polynomials given in succinct representation. For general circuits this makes no real difference
(since a big power $x^N$ can be defined by a subcircuit of size $O(\log N)$ using iterated squaring),
but for  skew circuits we will gain additional succinctness. We will use the term ``powerful skew circuits''.
Formally,  a {\em powerful skew circuit} over the polynomial ring $R[x_1, \ldots, x_k]$ is a skew circuit over
the ring $R[x_1, \ldots, x_k]$
as defined above, where the right-hand
side of every input gate is a polynomial that is given in succinct representation (equivalently, we could 
require that the right-hand side is a constant $a \in R$ or a power $x_i^N$ with $N$ given in binary notation).
We define the size of a powerful skew circuit $\cC$ as follows: First, define the size 
$\size_{\cC}(A)$ of a  gate $A \in V$ as follows: If 
$A$ is an addition gate or a multiplication gate, then $\size_{\cC}(A)=1$, and if 
$A$ is an input gate with  $\rhs(A) = p(x_1, \ldots, x_k)$, then 
$\size_{\cC}(A)= \ms p(x_1, \ldots, x_k) \ms$.
Finally, we define the size of $\cC$ as 
$\sum_{A \in V} \size_{\cC}(A)$. 

A {\em powerful branching program} is an algebraic branching program  $(V,E, \lambda, s,t)$ over a polynomial ring
$R[x_1, \ldots, x_k]$, where every edge label $\lambda(e)$ ($e \in E$) is a polynomial that is given in succinct representation.
The size of a powerful branching program is $\sum_{e \in E} \ms \lambda(e) \ms$.
From a given powerful skew circuit one can compute in logspace an equivalent powerful branching program
and vice versa.

Note that 
the transformation of a
powerful skew circuit over $R[x_1, \ldots, x_k]$  into an equivalent standard skew circuit
(where every input gate is labelled by a polynomial given in standard representation)
requires an exponential blow-up. For instance, the smallest standard skew circuit for the polynomial $x^N$ 
has size $N$, whereas $x^N$ can be trivially obtained by a powerful skew circuit of size $\lceil \log N \rceil$.

A central computational problem in computational algebra is {\em polynomial identity testing}, briefly PIT.
Let $R$ be a ring that is effective in the sense that elements of $R$ can be encoded by natural numbers
in such a way that addition and multiplication in $R$ become computable operations. 
Then, PIT for the ring $R$ is the following problem:

\medskip
\noindent
{\em Input:} A number $k \geq 1$ and a circuit $\cC$ over the ring $R[x_1, \ldots, x_k]$.

\noindent
{\em Question:} Is $\poly(\cC)$ the zero-polynomial?

\medskip
\noindent
For the rings $\mathbb{Z}$ and $\mathbb{Z}_p$ ($p$ prime) the following result
was shown in \cite{IbMo83}; for $\mathbb{Z}_n$ with $n$ composite, it was shown in 
\cite{AgrawalB03}.

\begin{theorem} \label{thm-ibarra}
For each of the rings $\mathbb{Z}$ and $\mathbb{Z}_n$ ($n \geq 2$), 
PIT belongs to the class $\coRP$.
\end{theorem}
Note that the number $k$ of variables is part of the input in PIT.
On the other hand, there is a well-known reduction from PIT to 
PIT restricted to univariate polynomials (polynomials with a single variable) \cite{AgrawalB03}.
For a multivariate polynomial $p(x_1,\ldots, x_k) \in R[x_1, \ldots, x_k]$ let 
$\deg_i(p)$ be the degree of $p$ in the variable $x_i$. It is the largest number $d$ such
that $x_i^d$ appears in a monomial of $p$.
Let $p(x_1,\ldots,x_k)$ be a polynomial and let 
$d=1+\max\{ \deg_i(p) \mid 1 \leq i \leq k \}$.
We define the univariate polynomial $\univ(p)$ as
$$
\univ(p) = p(y^1, y^{d}, \ldots, y^{d^{k-1}}) .
$$
Hence, the polynomial $\univ(p)$ is obtained from $p(x_1,\ldots,x_k)$ by replacing every
monomial $a \cdot x_1^{n_1} \cdots x_1^{n_k}$ by $a \cdot y^N$, where 
$N = n_1 + n_2 d + \cdots n_k d^{k-1}$ is the number with base-$d$ representation
$(n_1, n_2, \ldots, n_k)$. The polynomial $p$ is the zero-polynomial if and only if 
$\univ(p)$ is the zero-polynomial.

The following lemma can be also shown for arbitrary circuits, but we will only need it for powerful skew circuits.
\begin{lemma} \label{prop-make-univariate}
Given a powerful skew circuit $\cC$ for the polynomial $p(x_1,\ldots,x_k)$, the following can be be computed in $\NC^2$:
\begin{enumerate}[(i)]
\item  The binary encoding of $d=1+\max\{ \deg_i(p) \mid 1 \leq i \leq k \}$ and
\item a powerful skew circuit $\cC'$ for $\univ(p)$ .
\end{enumerate}
\end{lemma}

\begin{proof}
Let $\mathcal{C}$ be a powerful skew circuit for the  polynomial $p(x_1,\ldots,x_k)$.
In order to compute $\deg_i(p)$, we construct 
a circuit over the max-plus semiring 
as follows: Take the circuit $\cC$. If $A$ is an input gate
that is labelled with the polynomial $a(x_1, \ldots, x_k)$, then relabel
$A$ with the binary coded number $\deg_i(a)$. Moreover, for a gate $A$ with
$\rhs(A) = B + C$ (resp., $\rhs(A) = B \times C$) 
we set $\rhs(A) = \max(B, C)$ (resp., $\rhs(A) = B + C$). 
The resulting circuit is clearly skew. Therefore it can be evaluated in  $\NC^2$ by
Lemma~\ref{lemma-evaluate-skew}.

Once the number $d=1+\max\{ \deg_i(p) \mid 1 \leq i \leq k \}$ is computed we simply replace every monomial 
$a \cdot x_1^{n_1} \cdots x_k^{n_k}$ in the circuit $\cC$ by the monomial $a \cdot y^N$, where 
$N = n_1 + n_2 d + \cdots n_k d^{k-1}$. The binary encoding of $N$ can be computed
from the binary encodings of  $n_1, \ldots, n_k$ even in {\sf DLOGTIME}-uniform $\TC^0$.
\qed
\end{proof}
Note that the above reduction from multivariate to univariate circuits does
not work for standard skew circuits: the output circuit will be powerful skew even
if the input circuit is standard skew. For instance, the polynomial
$\prod_{i=1}^k x_i$ (which can be produced by a standard skew circuit of size $k$)
is transformed into the polynomial $y^{2^k-1}$, for which the smallest 
standard skew circuit has size $\Omega(2^k)$.

\section{PIT for powerful skew circuits}

The main result of this paper is:

\begin{theorem} \label{thm-PIT-coRNC}
For each of the rings $\mathbb{Z}$ and $\mathbb{F}_p$ ($p$ is a prime that can be part of the input in unary encoding), 
PIT for powerful skew circuits belongs to the class $\coRNC^2$.
\end{theorem}
The proof of Theorem~\ref{thm-PIT-coRNC} has two main ingredients: The randomized
identity testing algorithm of Agrawal and Biswas \cite{AgrawalB03} and the modular polynomial powering
algorithm of Fitch and Tompa \cite{FiTo85}. Let us start with the identity testing algorithm of  Agrawal and Biswas.
We will only need the version for the polynomial ring  $\mathbb{F}_p[x]$, where $p$ is a prime number.

Consider a polynomial $P(x) \in \mathbb{F}_p[x]$ of degree $d$.
The algorithm of Agrawal and  Biswas consists
of the following steps (later we will apply this algorithm to the polynomial defined by a powerful
skew circuit), where $0 < \epsilon < 1$ is an error parameter:
\begin{enumerate}
\item Let $\ell$ be a number with $\ell \geq \log d$ and $t = \max \{\ell, \frac{1}{\epsilon}\}$
\item Find the smallest prime number $r$ such that $r \neq p$ and $r$ does not divide
any of $p-1$, $p^2-1, \ldots, p^{\ell-1}-1$. It is argued in \cite{AgrawalB03} that $r \in O(\ell^2 \log p)$.
\item Randomly choose a tuple $b = (b_0, \ldots, b_{\ell-1}) \in \{0,1\}^\ell$ and compute the 
polynomial $T_{r,b,t}(x) = Q_r(A_{b,t}(x))$, where $Q_r(x) = \sum_{i=0}^{r-1} x^i$
is the $r^{\text{th}}$ cyclotomic polynomial and $A_{b,t}(x) = x^t + \sum_{i=0}^{\ell-1} b_i \cdot x^i$.
\item Accept, if $P(x) \text{ mod } T_{r,b,t}(x) = 0$, otherwise reject.
\end{enumerate}
Clearly, if $P(x) = 0$, then the above algorithm accepts with probability $1$. 
For a non-zero polynomial $P(x)$, Agrawal and Biswas proved:

\begin{theorem}[\cite{AgrawalB03}]
Let $P(x) \in \mathbb{F}_p[x]$ be a non-zero polynomial of degree $d$. 
The above algorithm rejects $P(x)$ with probability at least $\varepsilon$.
\end{theorem}
The second result we are using was shown by Fich and Tompa:
\begin{theorem}[\cite{FiTo85}]
The following computation can be done in $\NC^2$:   

\smallskip
\noindent
Input:  A unary encoded prime number $p$, polynomials $a(x), q(x) \in \mathbb{F}_p[x]$ such that
$\deg(a(x)) < \deg(q(x)) = d$, and a binary encoded number $N$.

\smallskip
\noindent
Output:  The polynomial $a(x)^N \text{ mod } q(x)$.
\end{theorem}

\begin{remark}
In \cite{FiTo85}, it is stated that the problem can be solved using circuits of depth
$(\log n)^2 \log \log n$ for the more general case that the underlying field is $\mathbb{F}_{p^\ell}$,
where $p$ and $\ell$ are given in unary representation. The main bottleneck is the computation
of an iterated matrix product $A_1 A_2 \cdots A_m$ (for $m$ polynomial in the input length) of $(d \times d)$-matrices
over the field $\mathbb{F}_{p^\ell}$. In our situation (where the field is $\mathbb{F}_p$)
we easily obtain an $\NC^2$-algorithm for this step: Two $(d \times d)$-matrices over 
$\mathbb{F}_p$ can be multiplied in 
$\NC^1$ (actually in $\mathsf{DLOGTIME}$-uniform $\TC^0$). Then we compute the product
$A_1 A_2 \cdots A_m$ by a balanced binary tree of depth $\log m$.
Also logspace-uniformity of the circuits is not stated explicitly in \cite{FiTo85}, but follows easily since only standard arithmetical operations
on binary coded numbers are used.
\end{remark}
{\em Proof of Theorem~\ref{thm-PIT-coRNC}.}
By Lemma~\ref{prop-make-univariate} we can restrict to univariate polynomials.
We first prove the theorem for the case of a powerful skew circuit
$\mathcal{C}$ over the field $\mathbb{F}_p$, where the prime number $p$ is part of the input
but specified in unary notation. 

Let $p$ be a unary encoded prime number and
$\mathcal{A} = (\{1, \ldots, n\}, 1, n, \lambda)$ be a
powerful branching program with $n$ nodes that is equivalent to $\mathcal{C}$. 
Let $P(x) = \poly(\mathcal{A}) \in \mathbb{F}_p[x]$.  Fix an error probability
$0 < \varepsilon < 1$. Our randomized $\NC^2$-algorithm is based on the 
identity testing algorithm of Agrawal and Biswas. It accepts with probability 1 
if $\poly(\mathcal{A})=0$ and accepts with probability at most $\epsilon$ if  $P(x) \neq 0$.
Let us go through the four steps of the Agrawal-Biswas algorithm to see that they can
be implemented in $\NC^2$.

\medskip
\noindent
{\em Step 1.} An upper bound on the degree of $P(x)$ can be computed in $\NC^2$ as in the proof
of Lemma~\ref{prop-make-univariate}.
For the number $\ell$ we can take the number of bits of this degree bound, which
is polynomial in the input size. Let $t = \max \{\ell, \frac{1}{\epsilon}\}$.

\medskip
\noindent
{\em Step 2.} 
For the prime number $r$ we know that $r \in O(\ell^2 \log p)$, which is a polynomial bound.
Hence, we can test in parallel all possible candidates for $r$. For a certain candidate $r$, we check
in parallel whether it is prime (recall that $r$ is of polynomial size) and whether it divides any of the numbers
$p-1$, $p^2-1, \ldots, p^{\ell-1}-1$. The whole computation is possible in $\NC^1$.

\medskip
\noindent
{\em Step 3.} 
Let $b = (b_0, \ldots, b_{\ell-1}) \in \{0,1\}^\ell$ be the chosen tuple. We have to compute the polynomial
$T_{r,b,t}(x) = Q_r(A_{b,t}(x))$, where $Q_r(x) = \sum_{i=0}^{r-1} x^i$ and $A_{b,t} = x^t + \sum_{i=0}^{\ell-1} b_i \cdot x^i$.
This is an instance of iterated multiplication (for the powers $A_{b,t}(x)^i$) and iterated addition of polynomials.
Hence, by Proposition~\ref{prop-eberly} also this step can be carried out in $\NC^1$.
Note that the degree of $T_{r,b,t}(x)$ is $t \cdot (r-1)$, i.e., polynomial in the input size.

\medskip
\noindent
{\em Step 4.} 
For the last step, we have to compute $P(x)$ mod $T_{r,b,t}(x)$. 
For this, we consider in parallel all monomials $a \cdot  x^N$  that occur in an edge label  of our powerful algebraic 
branching program $\mathcal{A}$. Recall that $a \in \mathbb{F}_p$ and $N$ is given in binary notation.
Using the Fich-Tompa algorithm we compute $x^N$ mod $T_{r,b,t}(x)$ (with $a(x) = x$) in $\NC^2$.
We then replace the edge label $a \cdot x^N$ by $a \cdot (x^N \text{ mod } T_{r,b,t}(x))$.
Let $\mathcal{B}$ the resulting algebraic branching program. Every polynomial that appears as an edge
label in $\mathcal{B}$ is now given in standard form. Hence, by Lemma~\ref{lemma-evaluate-skew} we 
can compute in $\NC^2$ the output polynomial $\poly(\mathcal{B})$. Clearly,
$P(x)$ mod $T_{r,b,t}(x) = \poly(\mathcal{B})$ mod $T_{r,b,t}(x)$. The latter polynomial can be computed
in $\NC^1$ by Proposition~\ref{prop-eberly}.

Let us now prove Theorem~\ref{thm-PIT-coRNC} for the ring $\mathbb{Z}$. 
Let $\mathcal{A} = (\{1, \ldots, n\}, 1, n, \lambda)$ be a powerful algebraic
branching program over $\mathbb{Z}$ with $n$ nodes and let $P(x) = \val(\mathcal{A})$.
Let us  first look at the coefficients of $P(x)$. 
Let $m$ be the maximum absolute value $|a|$, where $a \cdot  x^N$
is an edge label of $\mathcal{A}$. Since there are at most $2^n$ many paths from $s$ to $t$ in
$\mathcal{A}$, every coefficient of the polynomial $P(x)$ belongs to the interval
$[-(2m)^n,  (2m)^n]$. Let $k = n \cdot (\lceil \log(m) \rceil +1) +1$ and $p_1,\dots ,p_k$ be the first $k$ prime numbers.
Each prime $p_i$ is polynomially bounded in $k$ (and hence the input size) and the list of primes can be computed
in $\NC^1$ by doing in parallel all necessary divisibility checks on unary encoded numbers.

The Chinese remainder theorem implies that $P(x) = 0$ if and only if $P(x) \equiv 0$ mod $p_i$
for all $1 \leq i \leq k$. We can carry out the latter tests in parallel using the above algorithm 
for a unary encoded prime number. The overall algorithm accepts if we accept for every prime $p_i$.
If $P(x) = 0$, then we will accept for every  $1 \leq i \leq k$ with probability 1, hence the overall
algorithm accepts with probability 1. On the other hand, if $P(x) \neq 0$, then there exists  a prime $p_i$ ($1 \leq i \leq k$) such that the algorithm rejects
with probability at least $1-\varepsilon$. Hence, the overall algorithm will reject 
with probability at least $1-\varepsilon$ as well.
\qed

\section{Multi-dimensional straight-line programs} \label{sec-SLP}

Let $\Gamma$ be a finite alphabet. For $l \in \mathbb{N}$ 
let $[l] = \{1,\ldots, l\}$.
An {\em $n$-dimensional picture} over $\Gamma$ is a mapping $p:  \prod_{j=1}^n [l_j] \to \Gamma$ for some $l_j \in \mathbb{N}$.
Let $\dom(p) = \prod_{j=1}^n [l_j]$.
For $1 \leq j \leq n$ we define $|p|_j = l_j$ as the length of $p$ in the $j^{\text{th}}$ dimension. 
Note that one-dimensional pictures are simply finite words.
Let $\Gamma^*_n$ denote the set of $n$-dimensional pictures over $\Gamma$.
On this set we can define partially defined concatenation operations $\circ_i$ ($1 \leq i \leq n$)
as follows: For pictures $p, q \in \Gamma^*_n$, the picture $p \circ_i q$ is defined if and only
if $|p|_j = |q|_j$ for all $1 \leq j \leq n$ with $i \neq j$. In this case, we have
$|p \circ_i q|_j = |p|_j \,  (= |q|_j)$ for $j \neq i$ and $|p \circ_i q|_i = |p|_i +  |q|_i$.
Let $l_j = |p \circ_i q|_j$. For a tuple $(k_1, \ldots, k_n) \in \prod_{j=1}^n [l_j]$
we finally set $(p \circ_i q)(k_1, \ldots, k_n) = p(k_1, \ldots, k_n)$ if 
$k_i \leq  |p|_i$ and $(p \circ_i q)(k_1, \ldots, k_n) = q(k_1, \ldots,  k_{i-1}, k_i-|p|_i, k_{i+1}, \ldots, k_n)$
if $k_i >  |p|_i$. These operations generalize the concatenation of finite 
words.

An \emph{$n$-dimensional straight-line program (SLP)} over the terminal
alphabet $\Gamma$ is a triple $\dA = (V,\rhs,S)$, where $V$ is a finite set of variables,
$S \in V$ is the start variable, and $\rhs$ maps each variable $A$ to its right-hand
side $\rhs(A)$, which is either a terminal symbol $a \in \Gamma$ or an expression of 
the form $B \circ_i C$, where $B, C \in V$ and $1 \leq i \leq n$ such that the following additional 
conditions are satisfied:
\begin{itemize}
\item The relation $\{ (A,B) \in V \times V \mid A \text{ occurs in } \rhs(A) \}$ is acyclic.
\item  One can assign
to each $A \in V$ and $1 \leq i \leq n$ a number $|A|_i$ with the following properties: If $\rhs(A) \in \Gamma$ then 
$|A|_i = 1$ for all $i$. If $\rhs(A) = B \circ_i C$ then $|A|_i = |B|_i + |C|_i$ and
$|A|_j  = |B|_j = |C_j|$ for all $j \neq i$.
\end{itemize}
These conditions ensure that every variable $A$ evaluates to 
a unique $n$-dimensional picture $\val_\dA(A)$ such that
$|\val_\dA(A)|_i = |A|_i$ for all $1 \leq i \leq n$. Finally, $\val(\dA) = \val_{\dA}(S)$
is the picture defined by $\dA$. We omit the index $\dA$ if
the underlying SLP is clear from the context. We define 
the size of the SLP $\dA= (V,\Gamma,S,P)$ as $|\dA|=|V|$.

A one-dimensional SLP is a context-free grammar that generates a single word.
Two-dimensional SLPs were studied in \cite{BeKaLaPlRy02}. 

A simple induction shows that for every $n$-dimensional SLP $\dA$ of size $m$ and every $1 \leq i \leq n$ one has
$|\val(\dA)|_i \leq O(3^{m/3})$ \cite[proof of Lemma~1]{CLLLPPSS05}. On the other hand, it is straightforward to define an SLP
$\dB$ of size $m$ such that $|\val(\dB)|_i = 2^m$ for all dimensions $i$.
Hence, an SLP can be seen as a compressed representation of the picture it generates,
and exponential compression rates can be achieved in this way.

\subsection{Equality testing for compressed strings and n-dimensional pictures}

Given two $n$-dimensional SLPs we want to know whether they evaluate to the same picture.
In \cite{BeKaLaPlRy02} it was shown that this problem belongs to $\coRP$ by translating it to polynomial
identity testing. For a given $n$-dimensional picture $p: \dom(p) \to \{0,1\}$ we define the polynomial 
\[f_p(x_1,...,x_n)=\sum_{(e_1,...,e_n) \in \dom(p)} p(e_1,...,e_n) \prod_{i=1}^n x_i^{e_i} .\]
We consider $f_p$ as a polynomial from $\mathbb{Z}_2[x_1, \ldots, x_n]$. For two $n$-dimensional  pictures
  $p$ and $q$ such that $|p|_i = |q|_i$ for all $1 \leq i \leq n$ we clearly have
 $p = q$ if and only if $f_p + f_q = 0$ (recall that coefficients are from $\mathbb{Z}_2$).
 In \cite{BeKaLaPlRy02}, it was observed that from an SLP $\dA$ for a picture $P$, one can easily
 construct an arithmetic circuit for the polynomial $f_p$, which leads to a $\coRP$-algorithm
 for equality testing. 
 Since the circuit for $f_p$ is  actually powerful skew, we get:

\begin{theorem}
The question whether two $n$-dimensional SLPs $\mathcal{A}$ and $\mathcal{B}$ evaluate to the same $n$-dimensional picture is in $\coRNC^2$
(here, $n$ is part of the input).
\end{theorem}

\begin{proof}
Let $\mathbb{A}_1=(V_1,\rhs_1,S_1)$ and 
$\mathbb{A}_2=(V_2,\rhs_2,S_2)$ be $n$-dimensional SLPs over the alphabet
$\Gamma$. We can assume that $V_1 \cap V_2 = \emptyset$ and 
$\Gamma = \{0,1\}$ (if $\Gamma = \{a_1, \ldots, a_k\}$ then we encode 
$a_i$ by $0^i 1^{k-i}$).  

First we calculate $|A|_i$ for every $1 \leq k \leq n$ and every $A \in V_1 \cup V_2$ in $\NC^2$ by evaluating additive circuits
over $\mathbb{N}$, see Lemma~\ref{lemma-evaluate-skew}. If $|S_1|_i \neq |S_2|_i$ for some $1 \leq k \leq n$, then
we have $\val(\mathbb{A}_1)\neq\val(\mathbb{A}_2)$. Otherwise,
we construct the circuit 
$$\mathcal{C} = (V_1 \cup V_2 \cup \{S\}, \rhs, S)$$ over $\mathbb{Z}_2[x_1, \ldots, x_n]$
with:
\begin{gather*}
\rhs_{\mathcal{C}}(A)= B+ x_k^{|B|_k} \cdot C \text{ if }  \rhs_1(A)=B \circ_k C \text{ or } \rhs_2(A)=B \circ_k C, \\
  \rhs_{\mathcal{C}}(A)= a \text{ if } \rhs_1(A)=a \in \{0,1\}   \text{ or } \rhs_2(A)= a \in \{0,1\}, \text{ and } \\
\rhs(S)=S_1 + S_2
\end{gather*}
Then $\val(\mathcal{C})=f_{\val(\mathbb{A}_1)}+f_{\val(\mathbb{A}_2)}$ and so $\val(\mathcal{C})=0$ if and only if 
$\val(\mathbb{A}_1)=\val(\mathbb{A}_2)$. 
Obviously $\mathcal{C}$ becomes a powerful skew circuit after splitting right-hand sides of the form
$B + x^N \cdot C$. Hence, Theorem~\ref{thm-PIT-coRNC} allows to check in 
$\coRNC^2$ whether $\val(\mathcal{C})=0$.
\qed
\end{proof}
It should be noted that even in the one-dimensional case (where equality testing for SLPs can be 
done in polynomial time \cite{HiJeMo96,MehlhornSU97,Pla94}), no randomized $\NC$-algorithm was known before.

\section{Circuits over wreath products}

As a second application of identity testing for powerful skew circuits we will consider the circuit evaluation
problem (also known as the compressed word problem) for wreath products of finitely generated abelian groups.
The wreath product is an important operation in group theory. The next subsection briefly recalls the definition
and some well-known results. We assume some basic familiarity with group theory.

\subsection{Wreath products}

Let $G$ and $H$ be groups. The restricted wreath product
$H \wr G$ is defined as follows:
\begin{itemize}
\item 
Elements of $H \wr G$ are pairs $(f,g)$, where $g \in G$ and
$f : G \to H$ is a mapping such that $f(a) \neq 1_H$ for only finitely
many $a \in G$ ($1_H$ is the identity element of $H$).
\item 
The multiplication in $H \wr G$ is defined as follows:
Let $(f_1,g_1), (f_2,g_2) \in H \wr G$. Then 
$(f_1,g_1)(f_2,g_2) = (f, g_1g_2)$, where
$f(a) = f_1(a)f_2(g_1^{-1}a)$.
\end{itemize}
For readers, who have not seen this definition before, the following
intuition might be helpful: An element $(f,g) \in H \wr G$ can be thought
as a finite collection of elements of $H$ that are sitting in certain
elements of $G$ (the mapping $f$) together with a distinguished
element  of $G$ (the element $g$), which can be thought as a cursor
moving around $G$. 
If we want to compute the product $(f_1,g_1) (f_2,g_2)$, we do this 
as follows: First, we shift the finite collection of $H$-elements that
corresponds to the mapping $f_2$ by $g_1$: If the element $h \in H\setminus\{1_H\}$ is 
sitting in $a \in G$ (i.e., $f_2(a)=h$), then we remove $h$ from $a$ and
put it to the new location $g_1a \in G$. This new collection
corresponds to the mapping $f'_2 : a \mapsto f_2(g_1^{-1}a)$.
After this shift, we multiply the two collections of $H$-elements
pointwise: If in $a \in G$ the elements $h_1$ and $h_2$ are sitting
(i.e., $f_1(a)=h_1$ and $f'_2(a)=h_2)$, then we put the product
$h_1h_2$ into the $G$-location $a$. Finally, the new distinguished
$G$-element (the new cursor position) becomes $g_1g_2$. 

The following lemma seems to be folklore.

\begin{lemma} \label{prop-wreath-direct-product}
The group $(A \times B) \wr G$ embeds into 
$(A \wr G) \times (B \wr G)$.
\end{lemma}

\begin{proof}
Let $\pi_A : A \times B \to A$ be the natural projection morphism
and similarly for $\pi_B : A \times B \to B$.
We define an embedding $\varphi : (A \times B) \wr G \to
(A \wr G) \times (B \wr G)$ by
$$
\varphi(f,g) = \bigg( (f \circ \pi_A,g), (f \circ \pi_B,g) \bigg).
$$
Clearly, $\varphi$ is injective. Moreover, it is easy to see that
$\varphi$ is a group homomorphism.
\qed
\end{proof}
A proof of the  following simple lemma can be found for instance in \cite{LohreySZ14}.

\begin{lemma}\label{finiteindex}
Let $K$ be a subgroup of $H$ of finite index $m$ and let $G$ be a group.  Then $G^m\wr K$ is isomorphic to a subgroup of index $m$ in $G\wr H$.
\end{lemma}

\subsection{Compressed word problems} \label{sec-cwp}

Let $G$ be a finitely generated group and let $\Sigma$ be a finite generating set for $G$, i.e.,
every element of $G$ can be written as a finite product of elements from $\Sigma$ and inverses
of elements from $\Sigma$.
Let $\Gamma = \Sigma \cup \{ a^{-1} \mid a \in \Sigma\}$. For a word $w \in \Gamma^*$ we write
$w=1$ in $G$ if and only if the word $w$ evaluates to the identity of $G$. The {\em word problem}
for $G$ asks, whether $w=1$ in $G$ for a given input word. There exist finitely generated groups and 
in fact finitely presented groups (groups that are defined by finitely many defining relations) with an 
undecidable word problem. Here, we are interested in the {\em compressed word problem} for a finitely
generated group. For this, the input word $w$ is given in compressed form by a one-dimensional SLP as defined in Section~\ref{sec-SLP}.
Recall that a one-dimensional picture over an alphabet $\Gamma$ is simply a finite word over $\Gamma$.
Hence, $\val(\dA)$ is a word if $\dA$ is a one-dimensional SLP.
In the following we always mean one-dimensional SLPs when using the term SLP. 
The compressed word problem for $G$ asks, whether $\val(\dA)=1$ in $G$ for a given SLP $\dA$.

The compressed word problem is related to the classical word problem. For instance, the classical word problem for 
a f.g. subgroup of the automorphism group of a group $G$ can be reduced to the compressed word problem for $G$,
and similar results are known for certain group extensions, see \cite{Loh14} for more details. Groups, for which the compressed
word problem can be solved in polynomial time are  \cite{Loh14}: finite groups,  f.g. nilpotent groups, f.g. free groups,
graph groups (also known as right-angled Artin groups or partially commutative groups), and 
virtually special groups, which are groups that have a finite index subgroup that embeds into a graph group.
The latter groups form a rather large class that include for instance Coxeter groups, one-relator groups with torsion,
residually free groups, and fundamental groups of hyperbolic 3-manifolds.
In \cite{BeMcPeTh97} the parallel complexity of the compressed word problem (there, called the circuit evaluation problem) 
for finite groups was studied, and the following result was shown:

\begin{theorem}[\cite{BeMcPeTh97}] \label{thm-CWP-finite-groups}
Let $G$ be a finite group. If $G$ is solvable, then $\CWP(G)$ belongs to the class $\NC^2$.
If $G$ is not solvable, then $\CWP(G)$ is $\mathsf{P}$-complete.
\end{theorem}
The following two results are proven in \cite{Loh14}. 

\begin{theorem}[Theorem 4.15 in \cite{Loh14}] \label{thm-CWP-linear-groups}
For every f.g. linear group the compressed word problem belongs to the class $\mathsf{coRP}$.
\end{theorem}
This result is shown by reducing the compressed word problem for a f.g. linear group to polynomial identity testing 
for the ring $\mathbb{Z}$. Also a kind of converse of Theorem~ \ref{thm-CWP-linear-groups} is shown in \cite{Loh14}:

\begin{theorem}[Theorem 4.16 in \cite{Loh14}] \label{thm-CWP-SL3}
The problem $\CWP(\mathsf{SL}_3(\mathbb{Z}))$ and polynomial identity testing
for the ring $\mathbb{Z}$ are polynomial time reducible to each other.
\end{theorem}
This result is shown by using the construction of Ben-Or and Cleve \cite{Ben-OrC92} for simulating arithmetic circuits by matrix products.

Finally, the following result was recently shown in \cite{KoeLo15nil}; it generalizes Theorem~\ref{thm-CWP-finite-groups}.

\begin{theorem}[\cite{KoeLo15nil}] \label{thm-CWP-nilpotent}
Let $G$ be a f.g. group having a normal subgroup $H$ such that $H$ is f.g. nilpotent and 
the quotient group $G/H$ is finite solvable. Then $\CWP(G) \in \NC^2$.
\end{theorem}
To the knowledge of the author, there is no example of a group $G$ not having the properties from
Theorem~\ref{thm-CWP-nilpotent}, for which $\CWP(G)$ belongs to $\NC$.

\subsection{$\CWP(\mathbb{Z}  \wr \mathbb{Z})$ and identity testing for powerful skew circuits}
 
 In this section, we explore the relationship between the compressed word problem for 
 the wreath product  $\mathbb{Z}  \wr \mathbb{Z}$ and polynomial identity testing
 for powerful skew circuits. We show that these
 two problems are equivalent w.r.t. $\NC^2$-reductions. 
 
Let $G = \mathbb{Z} \wr \mathbb{Z}$. We consider the generators $a$ and $t$
of $G$, where $t$ (resp., $a$) generates the $\mathbb{Z}$ on the left (resp., right). 
So, with $a$ (resp., $a^{-1}$) we move the cursor to the left (resp., right) and
with $t$ (resp., $t^{-1}$) we add one (resp., subtract one) from the value at the 
current cursor position. Let $\Gamma = \{a,t,a^{-1}, t^{-1}\}$.

For a word $w \in \Gamma^*$  we
define  $\diff(w) = |w|_a - |w|_{a^{-1}} \in \mathbb{Z}$.  The word $w$ is \emph{positive} if 
$\diff(u)\geq 0$ for every prefix $u$ of $w$ that ends with $t$ or $t^{-1}$.
The word $w$ is \emph{well-formed}, if it is positive and  $\diff(w) = 0$.
If $w$ is positive and $(f,g) \in G$ is a group element represented by
the word $w$, then $f(x) \neq 0$ implies that $x \in \mathbb{N}$ (intuitively, the $\mathbb{Z}$-generator
$t$ or its inverse is never added to a position outside of $\mathbb{N}$). If in addition $w$ is well-formed
then $g = 0$.
For a given positive word $w \in \Gamma^*$ 
we define a polynomial $p_w(x) \in \mathbb{Z}[x]$ 
inductively as follows:
\begin{itemize}
\item $p_\varepsilon(x)=0.$
\item If $w = ua$ or $w = ua^{-1}$, then $p_w(x) = p_u(x)$.
\item If $w = ut^{\delta}$ with $\delta \in \{1,-1\}$, then $p_w(x) = p_u(x) + \delta \cdot x^{d}$, where
$d = \diff(w) = \diff(u)$. 
\end{itemize}
If the positive word $w$ represents the group element $(f,g) \in G$, then
the polynomial $p_w(x)$ encodes the mapping $f$ in the following sense:
The coefficient of the monomial $x^e$ in $p_w(x)$ is exactly
$f(e)$. In particular, the following equivalence holds for every positive word $w \in \Gamma^*$:
$$
w=1 \text{ in } G \quad \Leftrightarrow \quad (p_w(x) = 0 
\text{ and } \diff(w) = 0 )
$$

\begin{lemma} \label{lemma-SLP->circuit}
From a given SLP $\dA$ over the alphabet $\Gamma$ one
can compute in $\NC^2$ a 
powerful skew circuit $\cC$ 
such that 
$\poly(\cC) = p_{w}(x)$, where
$w =  a^k \,  \val(\dA) \,  a^{-k}$ and 
$k = |\val(\dA)|$.
In particular, $\val(\dA)=1$ in $G$ if and only if $(\poly(\cC)=0$ and $\diff(\val(\dA)) =  0)$.
\end{lemma}

\begin{proof} 
Let $k = |\val(\dA)|$. Our construction is divided into the following two
steps:

\medskip
\noindent
{\em Step 1.} 
Using iterated squaring, we add further nonterminals to $\dA$ such that 
$\val(A_k) = a^k$ and $\val(A_k^{-1}) = a^{-k}$
for distinguished non-terminals $A_k$ and $A_k^{-1}$. 
Then, we define the SLP $\dB$ by defining $\rhs_{\dB}(A) = A_k t^{\delta} A_k^{-1}$
for every variable $A$ with $\rhs_{\dA}(A) = t^{\delta}$  ($\delta \in \{-1,1\}$).
All other right-hand sides of $\dA$ are left unchanged.
Then, $\val(\dB) = a^k \, \val(\dA) \, a^{-k}$ 
in $G$. 

Let $\dB = (V,\rhs_{\dB},S)$ for the further consideration.
Note that for every $A \in V$,  the word $\val_{\dB}(A)$ is positive.
Hence, for every $A \in V$ we can define the polynomial $p_A(x) := p_{\val(A)}(x)$.
Moreover, let $d_{A} = \diff(\val(A)) \in \mathbb{Z}$; these numbers $d_{A}$ can be computed by an additive circuit in $\NC^2$, 
see Lemma~\ref{lemma-evaluate-skew}.

For every $A \in V$  let 
$$
m_{A} = \min(\{ \diff(u) \mid u \text{ is a prefix of $\val(A)$ that ends
  with $t$ or $t^{-1}$} \}),
$$
where we set $\min(\emptyset)=0$.
Since $\val(A)$ is positive, we have $m_{A} \geq 0$.
The polynomial $p_A(x)$ can be uniquely written as 
$$
p_A(x) = x^{m_A}  \cdot q_A(x),
$$
for a polynomial $q_A(x)$. 
The numbers $m_{A}$ can be computed
in $\NC^2$, using the following identity, where $\alpha(A)$ denotes the set of symbols occurring in $\val_{\dB}(A)$.
$$
m_{A} = \begin{cases} 
  0 & \text{ if } \rhs_{\dB}(A) = a^\delta \\
  k & \text{ if } \rhs_{\dB}(A) = A_k t^\delta A_k^{-1} \\
  \min \{m_{B}, d_{B}+m_{C}\} & \text{ if } \rhs_{\dB}(A) = BC \text{ and } \alpha(C) \cap \{t, t^{-1}\} \neq \emptyset \\
  m_{B}  & \text{ if } \rhs_{\dB}(A) = BC \text{ and } \alpha(C) \cap \{t, t^{-1}\} = \emptyset 
  \end{cases}
$$
Note that these rules define a skew circuit in the semiring $(\mathbb{Z} \cup \{\infty\}, \min, +)$. Hence, by
Lemma~\ref{lemma-evaluate-skew} the circuit can be evaluated in $\NC^2$.

\medskip
\noindent
{\em Step 2.}
We now construct a circuit $\cC$ such that
for every $A \in V$ we have:
$$
\poly_{\cC}(A) =  q_A(x) .
$$
We define the rules of the circuit $\cC$ as follows:
\begin{itemize}
\item If $\rhs_{\dB}(A) = a^\delta$ for $\delta \in \{-1,1\}$, then we set $\rhs_{\cal C}(A) = 0$.
\item If $\rhs_{\dB}(A) = A_k t^\delta A_k^{-1}$ for $\delta \in \{-1,1\}$, then we set $\rhs_{\cal C}(A) = \delta$.
\item If $\rhs_{\dB}(A) = BC$ and $\alpha(C) \cap \{t, t^{-1}\} = \emptyset$, then we set $\rhs_{\cal C}(A) = B$.
\item If $\rhs_{\dB}(A) = BC$ and $\alpha(C) \cap \{t, t^{-1}\} \neq \emptyset$, then $m_{A} = \min \{m_{B}, d_{B} +m_{C}\}$ and we set $\rhs_{\cal C}(A) = (M_B \times B) + (M_C \times C)$, where
\begin{eqnarray*}
M_B &=& \begin{cases} 1 & \text{ if } m_{B} \leq d_{B}+m_{C} \\
              x^{m_{B} - d_{B}-m_{C}} & \text{ if } m_{B} > d_{B}+m_{C} 
            \end{cases}     \\
M_C &=& \begin{cases} 1 & \text{ if } m_{B} \geq d_{B}+m_{C} \\
     x^{d_{B}+m_{C}-m_{B}}   & \text{ if } m_{B} < d_{B}+m_{C}   .
\end{cases}  
\end{eqnarray*}
\end{itemize}
Note that the resulting circuit is powerful skew.
\qed
\end{proof}

\begin{corollary} \label{corollary-SLP->circuit}
The compressed word problem for $\mathbb{Z} \wr \mathbb{Z}$ is $\NC^2$-reducible to 
PIT for powerful skew circuits over the ring $\mathbb{Z}[x]$.
\end{corollary}
In the rest of this section we show that PIT for powerful skew circuits
can be reduced in $\NC^2$ to $\CWP(\mathbb{Z} \wr \mathbb{Z})$.
By Proposition~\ref{prop-make-univariate}, it suffices to consider the univariate case.

\begin{lemma} \label{lemma-add}
Let  $u, v \in \Gamma^*$ be well-formed.
Then $w=uv$ is well-formed too and $p_w(x) = p_u(x) + p_v(x)$. 
\end{lemma}

\begin{lemma} \label{lemma-mult}
Let  $u \in \Gamma^*$ be well-formed, $n,m \in \mathbb{N}$ and let
$w = a^n u^m a^{-n}$.
Then $w$ is well-formed too and $p_w(x) = m \cdot x^n \cdot p_u(x) $. 
\end{lemma}

\begin{lemma} \label{lemma-circuit->SLP}
From a given powerful skew circuit $\cC$  over the ring $\mathbb{Z}[x]$, one can compute in $\NC^2$
an SLP $\dA$ over the alphabet $\Gamma$ such that the following holds:
\begin{itemize}
\item $\val(\dA)$ is well-formed and
\item $p_{\val(\dA)}(x) = \poly(\cC)$.
\end{itemize}
\end{lemma}

\begin{proof}
Let $\cC = (V, \rhs_{\cal C}, A_0)$. The set of nonterminals of our SLP $\dA$ contains $V$, 
a disjoint copy $V' = \{V' \mid A \in V\}$ of $V$,
plus some auxiliary nonterminals.
The start nonterminal is $A_0$.
For every nonterminal $A \in V$ we will have
$p_{\val_{\dA}(A)}(x) = \poly_{\cC}(A)$ and for every
nonterminal $A' \in V'$ we will have
$p_{\val_{\dA}(A')}(x) = -\poly_{\cC}(A)$.
We define the right-hand sides of $\dA$ as follows:
\begin{itemize}
\item If $\rhs_{\cal C}(A) = b \cdot x^n$, then we set $\rhs_{\dA}(A) = a^n t^b a^{-n}$ and $\rhs_{\dA}(A') = a^n t^{-b} a^{-n}$.
\item  If $\rhs_{\cal C}(A) = B+C$,  then we set $\rhs_{\dA}(A) = BC$ and $\rhs_{\dA}(A') = B'C'$
The correctness of this step follows from Lemma~\ref{lemma-add}. 
\item If $\rhs_{\cal C}(A) = B \times C$, where w.l.o.g. $C$ is an input gate with 
$\rhs_{\cal C}(C) = b \cdot x^n$, then we set   
$\rhs_{\dA}(A) = a^n B^{b} a^{-n}$ and $\rhs_{\dA}(A') =  a^n B^{-b} a^{-n}$,
where we set $B^{-x} = (B')^x$ for $x \geq 1$.
The correctness of this step
 follows from Lemma~\ref{lemma-mult}. 
\end{itemize}
It follows by a straightforward induction that for every $A \in V$, 
the strings $\val_{\dA}(A)$ and $\val_{\dA}(A')$ are well-formed.
\qed
\end{proof}
From Lemma~\ref{lemma-SLP->circuit} and \ref{lemma-circuit->SLP}
we directly obtain:

\begin{corollary}  \label{equiv1}
The compressed word problem for $\mathbb{Z}\wr\mathbb{Z}$ is equivalent w.r.t. $\NC^2$-reductions to 
PIT for powerful skew circuits over the ring $\mathbb{Z}[x]$.
\end{corollary}
In exactly the same way we can show:
\begin{corollary} \label{equiv2}
The compressed word problem for $\mathbb{Z}_n\wr\mathbb{Z}$ ($n \geq 2$) is equivalent w.r.t. $\NC^2$-reductions to 
PIT for powerful skew circuits over the ring $\mathbb{Z}_n[x]$.
\end{corollary}

\subsection{Compressed word problems in $\coRNC^2$}

In this section, we apply the results from the last section to find groups for which the compressed
word problem belongs to $\coRNC^2$. Recall from Section~\ref{sec-cwp} that the only known examples of groups with a word
problem in $\NC$ are groups $G$ having a normal subgroup $H$ such that (i) $H$ is f.g. nilpotent and (ii)
$G/H$ is finite solvable. For wreath products we use the following lemma:

\begin{lemma} \label{lemma-direct-product-on-right}
For every $k \geq 1$ and every finitely generated group $G$,  $\CWP(G \wr \mathbb{Z}^k)$ is $\NC^2$-reducible to 
 $\CWP(G \wr \mathbb{Z})$.
 \end{lemma}

\begin{proof}
The idea is similar to the proof of Proposition~\ref{prop-make-univariate}.
Let $G$ be generated by the finite set $\Sigma$.
 Fix the generating set $\{a_1, a_2, \ldots, a_k\}$ for $\mathbb{Z}^k$, where every
 $a_i$ generates a $\mathbb{Z}$-copy. Then $G \wr \mathbb{Z}^k$ is generated
 by the set $\Gamma = \Sigma \cup \{a_1, a_2, \ldots, a_k\}$. Let $\dA$ be an SLP over the alphabet
 $\Gamma \cup \Gamma^{-1}$. First, we compute in $\NC^2$ the number $d = 2(|\val(\dA)|+1)$.
Note that for all $a_i, b_i \in \mathbb{Z}$ ($1 \leq i \leq k$) with $|a_i|, |b_i| \leq |\val(\dA)|$ we have:
$(a_1, \ldots, a_k) = (b_1, \ldots, b_k)$ if and only if $\sum_{i=1}^k a_i \cdot d^{i-1} = \sum_{i=1}^k b_i \cdot d^{i-1}$.

From our SLP $\dA$ we construct a new SLP $\dB$ by replacing every occurrence of $a_i$ (resp., $a_i^{-1}$)
in a right-hand side by a new variable that produces $a^{d^{i-1}}$ (resp., $a^{-d^{i-1}}$). This implies the following:
If $(f, (z_1, \ldots, z_k))$ (resp., $(h, z)$) is the group element of  $\CWP(G \wr \mathbb{Z}^k)$ (resp., $\CWP(G \wr \mathbb{Z})$)
represented by $\val(\dA)$ (resp., $\val(\dB)$), then 
$z =  \sum_{i=1}^k z_i \cdot d^{i-1}$ and for all $(x_1, \ldots, x_k) \in \mathbb{Z}^k$,
$f(x_1, \ldots, x_k) = h(x)$, where $x =  \sum_{i=1}^k x_i \cdot d^{i-1}$. It follows that
$\val(\dA) = 1$ in $G \wr \mathbb{Z}^k$ if and only if $\val(\dB) = 1$ in $G \wr \mathbb{Z}$.
\qed
\end{proof}
By Lemma~\ref{prop-wreath-direct-product} and Lemma~\ref{lemma-direct-product-on-right}
the compressed word problem for a group $(G \times H)  \wr \mathbb{Z}^n$ can be reduced
in $\NC^2$ to the compressed word problems for the groups $G \wr \mathbb{Z}$ and $H \wr \mathbb{Z}$ 
Together with Theorem~\ref{thm-PIT-coRNC} and Corollary~\ref{equiv1} and \ref{equiv2} we obtain the following result:

\begin{corollary} \label{coro-coRNC-abelian}
Let $G$ be a finite direct product of copies of $\mathbb{Z}$ and $\mathbb{Z}_p$ for primes $p$.
Then, for every  $n \geq 1$, $\CWP(G \wr \mathbb{Z}^n)$ belong to $\coRNC^2$.
\end{corollary}
It is not clear, whether in Corollary~\ref{coro-coRNC-abelian} we can replace $G$ by an
arbitrary finitely generated abelian group. On the other hand,
if we apply Theorem~\ref{thm-ibarra} instead of Theorem~\ref{thm-PIT-coRNC} we obtain:

\begin{corollary} \label{coro-coRP-abelian}
Let $G$ be f.g. abelian and let $H$ be 
f.g. virtually abelian (i.e., $H$ has a f.g. abelian subgroup of finite index).
Then  $\CWP(G \wr H)$ belongs to $\coRP$.
\end{corollary}

\begin{proof}
Let $K \leq H$ be a f.g. abelian subgroup of finite index $m$ in $H$.
Moreover, either $K=1$ or $K \cong \mathbb{Z}^k$ for some $k \geq 1$.
By Lemma~\ref{finiteindex},  $G^m\wr K$ is isomorphic to a subgroup of index $m$ in $G\wr H$.
If the group $A$ is a finite index subgroup of the group $B$, then $\CWP(B)$ is polynomial-time many-one
reducible to $\CWP(A)$ \cite[Theorem~4.4]{Loh14}. 
Hence, it suffices to show that $\CWP(G^m\wr K)$ belongs to $\coRP$. Since $G^m$ is finitely generated abelian,
it suffices to consider $\CWP(\mathbb{Z}_n \wr K)$ ($n \geq 2$) and $\CWP(\mathbb{Z} \wr K)$. The case $K=1$ is clear.
So, assume that $K \cong \mathbb{Z}^k$. By Corollary~\ref{coro-coRNC-abelian}, 
$\CWP(\mathbb{Z} \wr \mathbb{Z}^k)$ belongs to $\coRNC$. Moreover, by Theorem~\ref{thm-ibarra}
and Corollary~\ref{equiv2}, $\CWP(\mathbb{Z}_n \wr \mathbb{Z}^k)$ belongs to $\coRP$.
\qed
\end{proof}
Recall that for a subgroup $H$ of a group $G$, $[H,H]$ denotes the {\em commutator subgroup} of $G$.
It is the subgroup of $G$ generated by all elements $h_1 h_2 h_1^{-1} h_2^{-1}$ with $h_1,h_2 \in H$.
It is well known that if $N$ is a normal subgroup of $G$, then also $[N,N]$ is a normal subgroup of $G$.
Hence, one can consider the quotient group $G/[N,N]$.
The following result of Magnus \cite{Mag39} has many applications in combinatorial group theory.

\begin{theorem}[Magnus embedding theorem]
Let $F_k$ be a free group of rank $k$ and 
let $N$ be a normal subgroup of $F_k$. Then 
$F_k/[N,N] \leq \mathbb{Z}^k \wr F_k/N$.
\end{theorem}

\begin{theorem}
Let $F_k$ be a free group of rank $k$ and 
let $N$ be a normal subgroup of $F_k$ such that $F_k/N$ is f.g. virtually abelian.
Then $\CWP(F_k/[N,N])$ belongs to $\coRNC^2$.
\end{theorem}

\begin{proof}
By the Magnus embedding theorem, the group $F_k/[N,N]$ embeds into the wreath product
$\mathbb{Z}^k \wr (F_k/N)$. For the latter group, the compressed word problem belongs 
to $\coRNC^2$ by Corollary~\ref{coro-coRNC-abelian}. 
\qed
\end{proof}

\section{Open problems}

Our $\coRNC^2$ identity testing algorithm for powerful skew circuits only works for the coefficient rings
$\mathbb{Z}$ and $\mathbb{Z}_p$ with $p$ prime. It is not clear how to extend it to $\mathbb{Z}_n$
with $n$ composite. The Agrawal-Biswas identity testing algorithm also works for  $\mathbb{Z}_n$
with $n$ composite.  But the problem is that the Fich-Tompa algorithm only works for polynomial rings
over $\mathbb{Z}_p$ with $p$ prime. For equality testing for multi-dimensional straight-line programs it 
remains open whether a polynomial time algorithm exists. For the one-dimensional (string) case, 
 a polynomial time algorithm exists. Here, it remains open, whether the equality problem is in $\NC$.

\def\cprime{$'$} \def\cprime{$'$}

\end{document}